\documentclass[letterpaper,12pt]{article}
\usepackage{amsfonts}
\usepackage{amsmath,amsthm}
\usepackage{amssymb,amscd}
\usepackage{enumerate}

\makeatletter

\newcommand{\Rmnum}[1]{\expandafter\@slowromancap\romannumeral #1@}
\makeatother

{}
\newtheorem{theorem}{Theorem}[section]

\newtheorem{lemma}[theorem]{Lemma}

\theoremstyle{definition}

\theoremstyle{remark}
\newtheorem{remark}{Remark}[section]

\swapnumbers \makeatletter
\def\ps@mystyles{    \def\@evenhead{\hfil{\sl\leftmark}\hfil}    \def\@oddhead{\hfil{\sl\rightmark}\hfil}    }
\makeatother

\begin{document}
\title{Distribution of a particle's position in the ASEP  with the {alternating} initial condition}

\author{Eunghyun Lee}

\maketitle

\begin{abstract}
  \noindent In this paper we give the distribution of the position of a particle in the asymmetric simple exclusion process (ASEP) with  the alternating initial condition. That is, we find $\mathbb{P}(X_m(t) \leq x)$ where $X_m(t)$ is the position of the particle at time $t$ which was at $m =2k-1, \hspace{0.1cm}k \in \mathbb{Z}$ at $t=0.$ As in the ASEP with step initial condition, there arises a new combinatorial identity for the alternating initial condition, and this identity relates the integrand of the integral formula for $\mathbb{P}(X_m(t) \leq x)$ to a determinantal form together with an extra product.

\end{abstract}

\section{Introduction}
\label{intro}
The exclusion process is an interacting stochastic particle system on a countable set $S$.  A particle at $x \in S$ chooses $y\in S$ with probability $p(x,y)$ after a holding time exponentially distributed with parameter 1. If  $y$ is empty, the particle at $x$ jumps to $y$ but if  $y$ is already occupied, then the particle remains at $x$, and the Poisson clock resumes. The detailed references on the construction of the model are Liggett's books \cite{Lig1,Lig2}. The asymmetric simple exclusion process (ASEP) is defined on $S=\mathbb{Z}$ by taking $p(x,x+1)=p$ and $p(x,x-1)=q$ for all $x \in \mathbb{Z}$, where $p+q=1$. If $p=1$, we call it the totally asymmetric simple exclusion process (TASEP). \\
\indent  Sch\"{u}tz \cite{Schut} considered the system of $N$ particles for the TASEP. There Sch\"{u}tz obtained the probability that the system  is in configuration $\{x_1,\cdots,x_N\}$ at time $t$  given the initial configuration $\{y_1,\cdots, y_N\} $ at $t=0$, and expressed the probability as an $N \times N$ determinant, and moreover, for general ASEP, obtained  the conditional  probability for $N=1,2$.
   Johansson \cite{Kurt} studied the TASEP with a special initial condition that one half of the system is occupied and the other half of the system is empty at $t=0$, which is called the \textit{step initial condition}.  Assuming the left half is occupied and the right half is empty at $t=0$, Johansson derived the probability that the $N$th particle from the rightmost moves at least $M$ steps before time $t$, which describes the \textit{time-integrated current}. By using the technique used in \cite{Schut} R\'{a}kos and Sch\"{u}tz \cite{Schut2} obtained the same result as Johansson \cite{Kurt} for the TASEP with  step initial condition.  R\'{a}kos and Sch\"{u}tz's method originates from the \textit{Bethe Ansatz} while Johansson used a combinatorial argument. \\
\indent   Tracy and Widom  extended the previous results on the TASEP to the  ASEP. Theorem 5.1 and Theorem 5.2 in \cite{TW1} provide  the probability that the $m$th particle from the leftmost is at $x \in \mathbb{Z}$ at time $t$ when the system initially has finitely many particles. Corollary to Theorem 5.2 is a generalization of the Johansson's result and the R\'{a}kos and Sch\"{u}tz's result to the ASEP and the subsequent \textit{Remark}  shows that how their result leads to the Johansson's result or the  R\'{a}kos and Sch\"{u}tz's result. The work for the deterministic step initial condition in \cite{TW1} was recently generalized to \textit{step Bernoulli initial condition} \cite{TW2} that assumes that at $t=0$ each site in $\mathbb{Z}^+$ is occupied with probability $\rho$ with $0 < \rho \leq 1$ independently of the others and all other sites are empty. \\
\indent In the ASEP with  step initial condition that the right half of the system is occupied and the left half is empty at $t=0$, the random variable $x_m(t)$, the $m$th particle's position from the leftmost at time $t$ can be related to the \textit{time-integrated current}. Assuming $p<q$,
\begin{equation*}
 \mathbb{P}(\mathcal{T}(x,t) \geq m) = \mathbb{P}(x_m(t) \leq x)
\end{equation*}
where $\mathcal{T}(x,t)$ is the number of particles whose positions are less than equal to $x$ at time $t$.  That the asymptotics on the fluctuation of $x_m(t)$ or $\mathcal{T}(x,t)$ are related to the GUE Tracy-Widom distribution in random matrix theory and that fluctuations are in the regime of the $t^{1/3}$ scale are well known for various situations \cite{Corwin,Kurt,Schut2,TW2}. \\
\indent Besides the step initial condition, the TASEP with the alternating initial condition also has been investigated.  This initial condition assumes that all even sites are occupied and all odd sites are empty at $t=0$ or vice-versa. While the current fluctuation of the TASEP (more generally ASEP) with step initial condition in the long limit is governed by the GUE Tracy-Widom distribution \cite{Corwin,Kurt,Schut2,TW5,TW2}, on the other hand, the current fluctuation of the TASEP with the alternating initial condition is related to the GOE Tracy-Widom distribution [3--6,10,15]. The appearance of the GOE Tracy-Widom distribution in growth models goes back to \cite{BaikRains} by Baik and Rains and a breakthrough in the TASEP was led by Sasamoto \cite{Sasamoto2}.  But to the best of the author's knowledge nothing is known on the ASEP with the alternating initial condition although we expect the GOE statistics in this case. \\
\indent In this paper, we study the distribution of a particle's position in the ASEP with the alternating initial condition, which gives information on the current of the system.  The asymptotic behavior of the current remains as a problem for the future. \\ \indent
We denote by $Y \subset \mathbb{Z}$ the set of initial positions of particles and by $\mathbb{P}_Y$ the probability of the ASEP with the initial condition $Y$.
 The main object we are interested in is a random variable $X_m(t)$, the position of a particle at time $t$ whose initial position is  $m\in Y$ and the main goal in this paper is to obtain $\mathbb{P}_Y(X_m(t) \leq x)$ when $Y=2\mathbb{Z}-1=\{2i-1 : i \in \mathbb{Z}\}.$   To do so we will start with a finite set $Y=\{2i-1: i \in \mathbb{Z},\hspace{0.2cm} -N+1 \leq i \leq N\}$ and will consider the limiting case that $N\rightarrow \infty$. Let $Y_+ = \{1,3,\cdots, 2N-1\}$ and $Y_- = \{-1,-3,\cdots, -2N+1\}$ so that $Y=Y_- \bigcup Y_+$. Let us denote the position of the $i$th particle from the leftmost at time $t$ by $x_i(t)$. Then,
\begin{equation}
   X_m(t) = x_{|Y_- | + \frac{m+1}{2}}(t), \label{def}
\end{equation}
so we can use some previous results on $x_i(t)$ in the ASEP with finite $Y$. In Section 2 we review the integral formula for $\mathbb{P}_Y(x_i(t) \leq x)$ when $Y$ is finite and in Section 3 by using a new combinatorial identity that arises in the alternating initial condition we derive $\mathbb{P}_{2\mathbb{Z}-1}(X_m(t) \leq x)$  and give the one-sided version of it. The identity is given in Lemma 3.2 and our final results for $\mathbb{P}_{2\mathbb{Z}-1}(X_m(t) \leq x)$ and its one-sided version are given in (\ref{final}) and (\ref{oneside}), respectively.

\section{Some known results}
\indent The integral formulas for the probability of a particle's position for the ASEP with a finite initial condition $Y$ were developed \cite{TW1,TW2,TW3}. These integral formulas have different forms depending on the contours we choose. The integral formula of Theorem 5.2 in \cite{TW1} which is over \textit{large} contours  can be used to derive the formula for step initial condition that positive integers are occupied as shown in Corollary. Alternatively, if  all negative integers are initially occupied and other sites are empty, the integral formula of Theorem 5.1 in \cite{TW1} which is over \textit{small} contours is needed. These are because we need geometric series which arise to be convergent.  Hence, for the alternating initial condition that has infinitely many particles on both sides of any reference site, a recently developed integral formula over both large and small contours is required \cite{TW3}.   In this section we review the formula developed in \cite{TW3} which is the starting point of this paper. \\
\indent First, let $\tau =p/q$ and $\xi=\xi(k)=(\xi_1,\cdots,\xi_{k_+},\xi_{-1},\cdots,\xi_{-k_-})$ with $k=k_++k_-$. We define
\begin{equation*}
 \varepsilon(\xi_i) := \frac{p}{\xi_i} + q\xi_i -1,\hspace{0.3cm} f(\xi_i,\xi_j) := \frac{\xi_j-\xi_i}{p+q\xi_i\xi_j -\xi_i}
\end{equation*}
{and}
\begin{equation*}
 I(x,\xi) := \prod_{i<j}f(\xi_i,\xi_j)\prod_i\frac{\xi_i^xe^{\varepsilon(\xi_i)t}}{1-\xi_i}.
\end{equation*}
Notice that $I(x,\xi)$ depends on $t$ but we omit it in the notation. Given two sets $U$ and $V$ of integers
\begin{equation*}
 \sigma(U,V) := \# \{(u,v) : u \in U,\hspace{0.1cm} v \in V \hspace{0.1cm}\textrm{ and }\hspace{0.1cm} u \geq v.\}
\end{equation*}
and recall the definition of the $\tau$-binomial coefficient,
\begin{equation*}
 \left[
   \begin{array}{c}
     N \\
     n \\
   \end{array}
 \right]_{\tau} = \frac{(1-\tau^N)(1-\tau^{N-1})\cdots(1-\tau^{N-n+1})}{(1-\tau)(1-\tau^2)\cdots(1-\tau^n)}.
\end{equation*}
We assume that $Y\subset \mathbb{Z} $ and $Y= Y_- \bigcup Y_+ $ is a finite set  where $Y_-$ and $Y_+$ are disjoint and all members of $Y_+$ are greater than all members in $Y_-$.
For $S_{\pm}\subset Y_{\pm}$ we set $|S_{\pm}| = k_{\pm}$ and use positive indices for $S_+$ and negative indices for $S_-$. In other words, we set $S_-=\{s_{-1},s_{-2}, \cdots , s_{-k_-}\}$ and $S_+=\{s_1,s_2, \cdots, s_{k_+}\}$.
  Then the distribution of $x_m(t)$, the $m$th particle's position from the leftmost particle at time $t$ is given by
\begin{equation}
\mathbb{P}_Y(x_m(t)\leq x) =\sum_{k_\pm\geq 0}{\sum_{\substack{S_\pm \subset Y_\pm,\\|S_\pm|=k_\pm}}}c_{m,S_-,S_+}\int_{\mathcal{C_R}^{k_+}}\int_{\mathcal{C}_r^{k_-}}
I(x,\xi)\prod_i\xi_i^{-s_i}\prod_id\xi_i \label{initial1}
\end{equation}
where
\begin{eqnarray*}
 c_{m,S_-,S_+} &=& (-1)^{m+|Y_-\setminus S_-|}\tau^{\frac{m(m-1)}{2}-k_+m+\sigma(S_+,Y)+\sigma(Y_-,Y_-\setminus S_-)-m|Y_-|+\frac{k_-(k_-+1)}{2}}\\
  & & \hspace{0.5cm}\times q^{\frac{k(k-1)}{2}}\left[
                                                 \begin{array}{c}
                                                   k-1 \\
                                                   m-|Y_-\setminus S_-|-1 \\
                                                 \end{array}
                                               \right]_{\tau}
\end{eqnarray*}
and $\mathcal{C_R}$ ($\mathcal{C}_r$) is a circle with center zero and radius $\mathcal{R}$ ($r$). Here $\mathcal{R}$ ($r$) is so large (small) that all the poles of $\prod_{i<j}f(\xi_i,\xi_j)$ in $I(x,\xi)$ lie inside (outside) $\mathcal{C_R}$ ($\mathcal{C}_r$).
The product in the integrand is over all positive and negative indices, and integrals over $\mathcal{C_R}$ are for variables of positive indices and integrals over $\mathcal{C}_r$ are for variables of negative indices. If we set $Y_-=\emptyset$, (\ref{initial1}) exactly becomes (3) in \cite{TW3}, from which we are able to work on the (positive) one-sided  step initial condition or the (positive) one-sided alternating initial condition.\\
\indent Additionally, in Section \Rmnum{5} of \cite{TW3},  a  variant of $c_{m,S_-,S_+}$ was computed. It is the coefficient \begin{eqnarray}
c_{m+|Y_- |,S_-,S_+} &=& (-1)^{m + k_-}\tau^{\sigma(S_+,Y_+\setminus S_+)-\sigma(Y_-\setminus S_-,S_-) + \frac{m(m-1)}{2}+\frac{k_+(k_++1)}{2}-mk_+}\nonumber \\
 & & \label{coef}\hspace{0.5cm}\times
q^{k(k-1)/2}\left[
                                                              \begin{array}{c}
                                                                k-1 \\
                                                               m + k_- -1 \\
                                                              \end{array}
                                                            \right]_\tau.
\end{eqnarray}
and will be used in the next section.
\section{Alternating initial condition on $\mathbb{Z}$}
\subsection{Distribution of a particle's position at time $t$}
In this section we derive $\mathbb{P}_{2\mathbb{Z}-1}(X_m(t) \leq x)$ as the limiting case $N \rightarrow \infty$ of the formula (\ref{initial1}) with $Y=\{-2N+1, -2N+3, \cdots, 2N-3,2N-1\}$ and with $ x_{|Y_- | + \frac{m+1}{2}}(t)$ instead of $x_m(t)$. Then $S_+$ and $S_-$, subsets of $Y_+$ and $Y_-$, respectively,  may be written as
 \begin{equation*}
  S_+ =\{2i_1-1, 2(i_1+i_2)-1, \cdots, 2(i_1+\cdots + i_{k_+}) -1\},\hspace{0.1cm} (i_1,\cdots,i_{k_+} \in \mathbb{N},
  \hspace{0.1cm} k_+ \leq N)
 \end{equation*}
 and
 \begin{equation*}
  S_- = \{-2j_1+1, -2(j_1+j_2)+1, \cdots, -2(j_1+\cdots + j_{k_-}) +1\},\hspace{0.05cm} (j_1,\cdots,j_{k_-} \in \mathbb{N}, \hspace{0.05cm}k_- \leq N).
 \end{equation*}
Since we are working with $x_{|Y_- | + \frac{m+1}{2}}(t)=X_m(t)$, the coefficient for $X_m(t)$ becomes
 \begin{eqnarray*}
c_{|Y_- | + \frac{m+1}{2},S_-,S_+} &=& (-1)^{\frac{m+1}{2} + k_-}\tau^{\big(\sigma(S_+,Y_+\setminus S_+) - \sigma(Y_-\setminus S_-,S_-) +\frac{m^2-1}{8}+\frac{k_+(k_++1)}{2}-\frac{m+1}{2}k_+\big)} \\
& & \hspace{1cm}\times
q^{k(k-1)/2}\left[
                                                              \begin{array}{c}
                                                                k-1 \\
                                                                \frac{m-1}{2} + k_- \\
                                                              \end{array}
                                                            \right]_\tau
\end{eqnarray*}
by replacing $m$ by $\frac{m+1}{2}$ in  (\ref{coef}) .
 Noticing that $\sigma(S_+,Y_+\setminus S_+) - \sigma(Y_-\setminus S_-,S_-)$ in $c_{|Y_- | + \frac{m+1}{2},S_-,S_+}$ and $\prod_i\xi_i^{-s_i}$ in the integrand in (\ref{initial1}) depend on $S_\pm$, we consider the sum
\begin{eqnarray*}
 & &\sum_{\substack{S_\pm \subset Y_\pm,\\ |S_\pm|=k_\pm}}\tau^{\sigma(S_+,Y_+\setminus S_+) - \sigma(Y_-\setminus S_-,S_-)}\prod_i\xi_i^{-s_i} \\
 &=&\sum_{\substack{S_+ \subset Y_+,\\ |S_+|=k_+}}\tau^{\sigma(S_+,Y_+\setminus S_+)}\prod_{i>0}\xi_i^{-s_i}\cdot
 \sum_{\substack{S_- \subset Y_-, \\|S_-|=k_-}}\tau^{-\sigma(Y_-\setminus S_-,S_-)}\prod_{i<0}\xi_i^{-s_i}.
\end{eqnarray*}
Let us compute the first sum. First,
\begin{eqnarray*}
  \prod_{i>0}\xi_i^{-s_i}& =& \xi_1^{-2i_1 +1}\xi_2^{-2i_1-2i_2 +1}\cdots\xi_{k_+}^{-2i_1-\cdots -2i_{k_+}+1} \\
   &=&(\xi_1\cdots\xi_{k_+})(\xi_1\cdots\xi_{k_+})^{-2i_1}(\xi_2 \cdots \xi_{k_+})^{-2i_2} \cdots \xi_{k_+}^{-2i_{k_+}}.
\end{eqnarray*}
Now, observe that the number of points in $Y_+$ less than or equal to $s_l=2(i_1+\cdots +i_l)-1$ is $i_1+\cdots +i_l$ and the number of points in $S_+$ less than or equal to $s_l$ is $l$.
Hence, the number of pairs $(s_l,y)$ with $s_l \geq y$ where $s_l \in S_+$ and $y \in Y_+ \setminus S_+$   is $i_1+\cdots + i_l -l$, and thus
\begin{eqnarray*}
 \sigma(S_+,Y_+\setminus S_+) &=& (i_1-1) + (i_1+i_2 -2) + \cdots +(i_1+\cdots + i_{k_+} - k_+) \\
  &=& k_+i_1 + (k_+-1)i_2 + \cdots + i_{k_+} - \frac{k_+(k_+ +1)}{2}.
\end{eqnarray*}
Denoting the first sum by $\varphi_+ ( k_+,\xi_+)$,
\begin{eqnarray*}
\varphi_+ ( k_+,\xi_+)&=&\sum_{\substack{S_+ \subset Y_+,\\ |S_+|=k_+}}\tau^{\sigma(S_+,Y_+\setminus S_+)}\prod_{i>0}\xi_i^{-s_i}\\
&=&\frac{\xi_1\cdots\xi_{k_+}}{\tau^{\frac{k_+(k_+ +1)}{2}}}\sum_{\substack{S_+ \subset Y_+,\\ |S_+|=k_+}} \Big(\frac{\tau^{k_+}}{(\xi_1\cdots\xi_{k_+})^2}\Big)^{i_1}\Big(\frac{\tau^{k_+-1}}{(\xi_2\cdots\xi_k)^2}\Big)^{i_2}
\cdots\Big(\frac{\tau}{\xi_{k_+}^2}\Big)^{i_{k_+}}.
\end{eqnarray*}
If we assume that $Y_+ =\{2i-1: i =1,2,\cdots\}$, the sum implies  geometric series, which converge because we choose large contours for variables with positive indices as shown in (\ref{initial1}). Hence,
\begin{eqnarray*}
\varphi_+ ( k_+,\xi_+)&=&\frac{\xi_1\cdots\xi_{k_+}}{\tau^{\frac{k_+(k_+ +1)}{2}}}\sum_{i_1,\cdots,i_{k_+}=1}^{\infty} \Big(\frac{\tau^{k_+}}{(\xi_1\cdots\xi_{k_+})^2}\Big)^{i_1}\Big(\frac{\tau^{k_+-1}}{(\xi_2\cdots\xi_{k_+})^2}\Big)^{i_2}
\cdots\Big(\frac{\tau}{\xi_{k_+}^2}\Big)^{i_{k_+}} \\
&=&\frac{\xi_1 \cdots \xi_{k_+}}{\big((\xi_1\cdots\xi_{k_+})^2 - \tau^{k_+}\big)\big((\xi_2\cdots\xi_{k_+})^2 - \tau^{k_+-1}\big)\cdots \big(\xi_{k_+}^2 -\tau\big)}.
\end{eqnarray*}
Likewise for negative indices,
\begin{eqnarray*}
  \prod_{i<0}\xi_i^{-s_i}& =& \xi_{-1}^{2j_1 -1}\xi_{-2}^{2j_1+2j_2 -1}\cdots\xi_{-k_-}^{2j_1+\cdots +2j_{k_-}-1} \\
   &=&(\xi_{-1}\cdots\xi_{-k_-})^{-1}(\xi_{-1}\cdots\xi_{-k_-})^{2j_1}(\xi_{-2} \cdots \xi_{-k_-})^{2j_2} \cdots \xi_{-k_-}^{2j_{k_-}}.
\end{eqnarray*}
For ${\sigma(Y_-\setminus S_-,S_-)}$ we set
\begin{equation*}
 \tilde{s}_l = -s_{-l} ,\hspace{0.2cm} \tilde{S}_+ = -S_-\hspace{0.2cm}\textrm{and}\hspace{0.2cm} \tilde{Y}_+ = -Y_-,
\end{equation*}
and then it is easily seen that
\begin{equation*}
 \sigma(Y_-\setminus S_-,S_-) = \sigma(\tilde{S}_+,\tilde{Y}_+\setminus \tilde{S}_+)
\end{equation*}
and so
\begin{eqnarray*}
 \sigma(\tilde{S}_+,\tilde{Y}_+\setminus \tilde{S}_+) &=& (j_1-1) + (j_1+j_2 -2) + \cdots +(j_1+\cdots +j_{k_-} - k_-) \\
  &=& k_-j_1 + (k_--1)j_2 + \cdots + j_{k_-} - \frac{k_-(k_- +1)}{2}.
\end{eqnarray*}
 Recalling that we choose small contours for negative indices so that geometric series for negative indices converge, one can obtain for $Y_-=\{2i-1: i=0,-1,-2\cdots\}$
\begin{eqnarray*}
& &\sum_{\substack{S_- \subset Y_-,\\ |S_-|=k_-}}\tau^{-\sigma(Y_-\setminus S_-,S_-)}\prod_{i<0}\xi_i^{-s_i} \\
  &=&\frac{\tau^{\frac{k_-(k_- +1)}{2}}}{\xi_{-1}\cdots\xi_{-k_-}}\sum_{j_{1},\cdots,j_{k_-}=1}^{\infty}
  \Big(\frac{(\xi_{-1}\cdots\xi_{-k_-})^2}{\tau^{k_-}}\Big)^{j_1}\Big(\frac{(\xi_{-2}\cdots\xi_{-k_-})^2}{\tau^{k_--1}}\Big)^{j_2}
\cdots\Big(\frac{\xi_{-k_-}^2}{\tau}\Big)^{j_{k_-}}
  \\
 &=&\tau^{\frac{k_-(k_-+1)}{2}}\cdot\frac{\xi_{-1}^1\xi_{-2}^3 \cdots \xi_{-k_-}^{2k_--1}}{\big((\tau^{k_-} - (\xi_{-1}\cdots\xi_{-k_-})^2 \big)\big(\tau^{k_--1}-(\xi_{-2}\cdots\xi_{-k_-})^2\big)\cdots \big(\tau - \xi_{-k_-}^2 \big)},
\end{eqnarray*}
and let
\begin{eqnarray*}
 & &\varphi_-(k_-,\xi_-)\\
 &:=& \frac{\xi_{-1}^1\xi_{-2}^3 \cdots \xi_{-k_-}^{2k_--1}}{\big((\tau^{k_-} - (\xi_{-1}\cdots\xi_{-k_-})^2 \big)\big(\tau^{k_--1}-(\xi_{-2}\cdots\xi_{-k_-})^2\big)
 \cdots \big(\tau - \xi_{-k_-}^2 \big)}.
\end{eqnarray*}
Hence we obtained
\begin{equation}
 \mathbb{P}_{2\mathbb{Z}-1}(X_m(t) \leq x)
 =\sum_{k_\pm\geq 0}c_{m,k_{\pm}}\int_{\mathcal{C_R}^{k_+}}\int_{\mathcal{C}_r^{k_-}}
I(x,\xi)\varphi_- ( k_-,\xi_-)\varphi_+ ( k_+,\xi_+)\prod_id\xi_i \label{form}
\end{equation}
where
\begin{equation}
  c_{m,k_{\pm}} = (-1)^{\frac{m+1}{2}+k_-}\tau^{\big(\frac{m^2-1}{8} + \frac{k(k+1)}{2}-k_+k_- - \frac{m+1}{2}k_+\big)}q^{\frac{k(k-1)}{2}}
  \left[
                                                              \begin{array}{c}
                                                                k-1 \\
                                                                \frac{m-1}{2} + k_- \\
                                                              \end{array}
                                                            \right]_\tau.
\end{equation}
\begin{remark}
 In Corollary in \cite{TW1} the identity (1.7) in \cite{TW1} was used in obtaining $\mathbb{P}_{\mathbb{Z}^+}(x_m(t)=x)$ for step initial condition. We give the $\mathbb{P}_{\mathbb{Z}^+}(x_m(t)\leq x)$ to be compared with the case of the alternating initial condition in the later section.
 \begin{eqnarray}
  \mathbb{P}_{\mathbb{Z}^+}(x_m(t) \leq x) &=& \label{step}(-1)^m\sum_{k \geq m}\frac{\tau^{(k-m)(k-m+1)/2}}{(1+\tau)^{k(k-1)}k!}
  \left[
                                                                                                             \begin{array}{c}
                                                                                                               k-1 \\
                                                                                                               k-m \\
                                                                                                             \end{array}
                                                                                                           \right]_{\tau}  \\
      & & \nonumber \hspace{0.2cm}\times\int_{\mathcal{C_{R}}}\cdots  \int_{\mathcal{C_{R}}}  \prod_{i\neq j}\frac{\xi_j-\xi_i}{p+q\xi_i\xi_j -\xi_i}\prod_{i}\frac{\xi_i^xe^{t\varepsilon (\xi_i)}}{(1-\xi_i)(\xi_i-\tau)}\prod_id\xi_i.
  \end{eqnarray}
  Moreover, the integrand  in  $\mathbb{P}_{\mathbb{Z}^+}(x_m(t) \leq x)$ could be expressed as a determinantal form by using another identity (3) in \cite{TW4}. It states that
  \begin{eqnarray}
   \mathbb{P}_{\mathbb{Z}^+}(x_m(t) \leq x) &=&  \sum_{k \geq m} c_{m,k,\tau}\int_{\mathcal{C_{R}}}\cdots  \int_{\mathcal{C_{R}}} \det(K(\xi_i,\xi_j))_{1\leq i,j \leq k}
   \prod_id\xi_i
  \end{eqnarray}
  where
  \begin{equation}
   K(\xi,\xi') =\frac{\xi^x e^{\varepsilon(\xi)t}}{p+q\xi\xi'- \xi} \label{operator}
  \end{equation}
  and $c_{m,k,\tau}$ is a constant depending on $m,k,$ and $\tau$. Here we introduce the identity for the later use.
  \end{remark}
  \begin{lemma}\cite{TW4}
   \begin{equation*}
 \det \Big(\frac{1}{p+q\xi_i\xi_j -\xi_i}\Big)_{1\leq i,j \leq k} = (-1)^k(pq)^{\frac{k(k-1)}{2}}q^{-k}\prod_{i \neq j}\frac{\xi_j -\xi_i}{p+q\xi_i\xi_j - \xi_i}\prod_i\frac{1}{(1-\xi_i)(\xi_i-\tau)}.\label{identity1}
\end{equation*}
\end{lemma}
\subsection{Symmetrization and Combinatorial identity}
As mentioned in \textit{Remark 3.1} a combinatorial identity was found to derive the integral formula of the distribution in case of step initial condition and the integrand of the formula can be expressed as a determinant. This identity is associated with a special initial structure of the system, that is, the step initial condition. So we may expect to have a new identity associated with the alternating initial condition. In this subsection we find the new combinatorial identity\footnote{This identity was conjectured by Craig A. Tracy through private communication.} and obtain an alternate form of (\ref{form}) by using the identity.

\begin{lemma}
Let $\tau= \frac{p}{q}$ and $p+q=1.$ For $k \in \mathbb{N}$
\begin{eqnarray*}
& &\hspace{0.3cm}\sum_{\sigma \in \mathbb{S}_k}\prod_{i>j}\frac{p+q\xi_{\sigma(i)}\xi_{\sigma(j)}-\xi_{\sigma(i)}}{\xi_{\sigma(j)} - \xi_{\sigma(i)}} \times \\
& &\hspace{0.3cm}\frac{1}{(\xi_{\sigma(1)}^2\xi_{\sigma(2)}^2\cdots\xi_{\sigma(k)}^2-\tau^k)
(\xi_{\sigma(2)}^2\xi_{\sigma(3)}^2\cdots\xi_{\sigma(k)}^2-\tau^{k-1})\cdots(\xi_{\sigma(k)}^2-\tau)}\\
&= &  \frac{1}{(1+\tau)^{k(k-1)/2}}\prod_{i<j}\frac{1+\tau - (\xi_i + \xi_j)}{\tau-\xi_i\xi_j}\prod_i\frac{1}{\xi_i^2-\tau}.
\end{eqnarray*}
\end{lemma}
\begin{proof}
The equality clearly holds for $k=1$. Denote the left hand side by $L_k(\xi_1,\cdots,\xi_k)$ and the right hand side by $R_k(\xi_1,\cdots, \xi_k)$, and assume that the identity holds for $k-1$, i.e, $L_{k-1}=R_{k-1}$. Let $\sigma(1) =l$. We change the sum over all permutations in $L_k$ to the double sum over $l=1,2,\cdots,k$ and $(\sigma(2), \cdots, \sigma(k)) \in \mathbb{S}_{k-1}$, that is,  $\sum_{\sigma \in \mathbb{S}_k}=\sum_{l=1}^k\sum_{(\sigma(2),\cdots,\sigma(k)) \in \mathbb{S}_{k-1}}$. Observe that $$
 \xi_{\sigma(1)}^2\xi_{\sigma(2)}^2\cdots\xi_{\sigma(k)}^2-\tau^k = \xi_{1}^2\xi_{2}^2\cdots\xi_{k}^2-\tau^k
$$ for all $\sigma \in \mathbb{S}_k$ and
\begin{eqnarray*}
 & &\prod_{\substack{
      i>j, \\
      i,j=1,\cdots, k}}
 \frac{p+q\xi_{\sigma(i)}\xi_{\sigma(j)}-\xi_{\sigma(i)}}{\xi_{\sigma(j)} - \xi_{\sigma(i)}} \\ &=&\prod_{i=2}^k\frac{p+q\xi_{\sigma(i)}\xi_{\sigma(1)}-\xi_{\sigma(i)}}{\xi_{\sigma(1)} - \xi_{\sigma(i)}}
 \prod_{\substack{
      i>j, \\
      i,j=2,\cdots, k}}
  \frac{p+q\xi_{\sigma(i)}\xi_{\sigma(j)}-\xi_{\sigma(i)}}{\xi_{\sigma(j)} - \xi_{\sigma(i)}} \\
  &=& \prod_{i\neq l}\frac{p+q\xi_{i}\xi_{l}-\xi_{i}}{\xi_{l} - \xi_{i}}
 \prod_{\substack{
      i>j, \\
      i,j=2,\cdots, k}}
  \frac{p+q\xi_{\sigma(i)}\xi_{\sigma(j)}-\xi_{\sigma(i)}}{\xi_{\sigma(j)} - \xi_{\sigma(i)}}
\end{eqnarray*}
and
\begin{eqnarray*}
& &L_{k-1}(\xi_1,\cdots,\xi_{l-1},\xi_{l+1},\cdots,\xi_k)=\sum_{\substack{(\sigma(2),\cdots,\sigma(k))  \\ \in \mathbb{S}_{k-1}}}\prod_{\substack{
      i>j, \\
      i,j=2,\cdots, k}} \frac{p+q\xi_{\sigma(i)}\xi_{\sigma(j)}-\xi_{\sigma(i)}}{\xi_{\sigma(j)} - \xi_{\sigma(i)}}  \\
  & &\hspace{2cm} \times\frac{1}{(\xi_{\sigma(2)}^2\cdots\xi_{\sigma(k)}^2-\tau^k)
(\xi_{\sigma(3)}^2\cdots\xi_{\sigma(k)}^2-\tau^{k-1})\cdots(\xi_{\sigma(k)}^2-\tau)}.
\end{eqnarray*}
Hence
\begin{eqnarray*}
L_k(\xi_1,\cdots, \xi_k) &=&\frac{1}{\xi_1^2\xi_2^2\cdots\xi_k^2-\tau^k}\sum_{l=1}^k\prod_{i\neq l}\frac{p+q\xi_{i}\xi_{l}-\xi_{i}}{\xi_{l} - \xi_{i}}L_{k-1}(\xi_1,\cdots,\xi_{l-1},\xi_{l+1},\cdots,\xi_k)\\
&=& \frac{1}{\xi_1^2\xi_2^2\cdots\xi_k^2-\tau^k}\sum_{l=1}^k\prod_{i\neq l}\frac{p+q\xi_{i}\xi_{l}-\xi_{i}}{\xi_{l} - \xi_{i}}R_{k-1}(\xi_1,\cdots,\xi_{l-1},\xi_{l+1},\cdots,\xi_k),
\end{eqnarray*}
where the second equality comes from the induction hypothesis. Our goal is to show that
\begin{equation*}
 \frac{1}{\xi_1^2\xi_2^2\cdots\xi_k^2-\tau^k}\sum_{l=1}^k\prod_{i\neq l}\frac{p+q\xi_{i}\xi_{l}-\xi_{i}}{\xi_{l} - \xi_{i}}R_{k-1}(\xi_1,\cdots,\xi_{l-1},\xi_{l+1},\cdots,\xi_k) =R_k(\xi_1,\cdots,\xi_k)
\end{equation*}
i.e.,
\begin{equation*}
 \sum_{l=1}^k\prod_{i\neq l}\frac{p+q\xi_{i}\xi_{l}-\xi_{i}}{\xi_{l} - \xi_{i}}\frac{R_{k-1}(\xi_1,\cdots,\xi_{l-1},\xi_{l+1},\cdots,\xi_k)}{R_k(\xi_1,\cdots,\xi_k)} =\xi_1^2\xi_2^2\cdots\xi_k^2-\tau^k.
\end{equation*}
Recalling the form of $R_k$, what we want to show is
\begin{equation}
 \frac{\xi_1^2\xi_2^2\cdots\xi_k^2-\tau^k}{(1+\tau)^{k-1}}=\sum_{l=1}^k(\xi_l^2-\tau)\prod_{i\neq l}\frac{p+q\xi_{i}\xi_{l}-\xi_{i}}{\xi_{l} - \xi_{i}}\frac{\tau-\xi_l\xi_i}{1+\tau-\xi_l-\xi_i}. \label{ident1}
\end{equation}
For some technical reasons we multiply by $(p+q\xi_l^2-\xi_l)(1+\tau-2\xi_l)$ both the numerator and the denominator of the right hand side, and then using $p+q\xi_l^2-\xi_l=(q\xi_l-p)(\xi_l-1)$, (\ref{ident1})  becomes
\begin{equation}
  \frac{\xi_1^2\xi_2^2\cdots\xi_k^2-\tau^k}{(1+\tau)^{k-1}}=\sum_{l=1}^k\frac{\prod_ip+q\xi_{i}\xi_{l}-\xi_{i}}{\prod_{i\neq l}\xi_{l} - \xi_{i}}\prod_i\frac{\tau-\xi_l\xi_i}{1+\tau-\xi_l-\xi_i}\frac{2\xi_l-1-\tau}{(\xi_l-1)(q\xi_l-p)}. \label{ident2}
\end{equation}
Observe that
\begin{equation*}
g(z):=\prod_l\frac{p+q\xi_{l}z-\xi_{l}}{z - \xi_{l}}\prod_l\frac{\tau-z\xi_l}{1+\tau-z-\xi_l}\frac{2z-1-\tau}{(z-1)(qz-p)}
 \sim\frac{2q^{k-1}\prod_l\xi_l^2}{z}
\end{equation*}
    for $z \rightarrow \infty$.
    Thus, the integral of $g(z)$ over a circle with sufficiently large radius $\mathcal{R}$ is equal to the integral of $\frac{2q^{k-1}\prod_l\xi_l^2}{z}$ over the circle, so,
      $$\int_{\mathcal{C_R}} g(z) dz =2q^{k-1}\prod_l\xi_l^2 =\sum\hspace{0.05cm} \mathbf{Res}\hspace{0.05cm}g(z).$$
      Recall that $\tau=p/q$ and $p+q=1$. It is easy to see that $\mathbf{Res}_{z=1}\hspace{0.05cm}g(z)=p^k/q$ and $\mathbf{Res}_{z=\frac{p}{q}}\hspace{0.05cm}g(z)=p^k/q$. Using $p+q\xi_l^2-\xi_l=(q\xi_l-p)(\xi_l-1)$ we obtain
\begin{equation*}
 \mathbf{Res}_{z=\xi_l} g(z)=(\xi_l^2-\tau)\prod_{i\neq l}\frac{p+q\xi_{i}\xi_{l}-\xi_{i}}{\xi_{l} - \xi_{i}}\frac{\tau-\xi_l\xi_i}{1+\tau-\xi_l-\xi_i}.
\end{equation*}
Finally, observing that $1+\tau =1/q$, we can see that  $\displaystyle\mathbf{Res}_{z=1+\tau -\xi_l} g(z) = \mathbf{Res}_{z=\xi_l} g(z)$, and thus (\ref{ident1}) is verified. This completes the proof.\qed
\end{proof}
We use this new identity to symmetrize the integrand in (\ref{form}). First, we write the integrand as
 \begin{eqnarray*}
I(x,\xi)\varphi_+ ( k_+,\xi_+)\varphi_- ( k_-,\xi_-)& =& \prod_{i<j}f(\xi_i, \xi_j)\prod_i\frac{\xi_i^xe^{\varepsilon(\xi_i)t}}{1-\xi_i}\varphi_+ ( k_+,\xi_+)\varphi_- ( k_-,\xi_-) \\
&=&\prod_{i<0,j>0}f(\xi_i,\xi_j)\Big[\prod_{0<i<j}f(\xi_i, \xi_j)\prod_{i>0}\frac{\xi_i^xe^{\varepsilon(\xi_i)t}}{1-\xi_i}\varphi_+ ( k_+,\xi_+)\Big] \\
& &\times
\Big[\prod_{i<j<0}f(\xi_i, \xi_j)\prod_{i<0}\frac{\xi_i^xe^{\varepsilon(\xi_i)t}}{1-\xi_i}\varphi_- ( k_-,\xi_-)\Big] . \label{symm}
\end{eqnarray*}
The first bracket is a function of variables with positive indices that is given by
\begin{eqnarray*}
 & & \nonumber\prod_{0<i<j}f(\xi_i, \xi_j)\prod_{i>0}^{k_+}\frac{\xi_i^xe^{\varepsilon(\xi_i)t}}{1-\xi_i}\varphi_+ ( k_+,\xi_+) =\nonumber \prod_{0<i \neq j}\frac{\xi_j - \xi_i}{p+q\xi_i\xi_j -\xi_i}\prod_{i>0}^{k_+}\frac{\xi_i^{x+1}e^{\varepsilon(\xi_i)t}}{1-\xi_i}\\
 & & \hspace{0.3cm}\times \prod_{0<j<i}\frac{p+q\xi_i\xi_j-\xi_i}{\xi_j-\xi_i} \nonumber\frac{1}{(\xi_{1}^2\xi_{2}^2\cdots\xi_{k_+}^2-\tau^{k_+})
(\xi_{2}^2\xi_{3}^2\cdots\xi_{k_+}^2-\tau^{k_+-1})\cdots(\xi_{k_+}^2-\tau)}
\end{eqnarray*}
and its symmetrization by the identity in Lemma 3.2 is
\begin{equation}
  c_{k_+}\prod_{0<i \neq j}\frac{\xi_j - \xi_i}{p+q\xi_i\xi_j -\xi_i}\prod_{i>0}\Big(\frac{\xi_i^{x+1}e^{\varepsilon(\xi_i)t}}{1-\xi_i}\frac{1}{\xi_i^2-\tau}\Big)
 \prod_{0<i<j}\frac{1+\tau -(\xi_i+\xi_j)}{\tau -\xi_i\xi_j} \label{symm2}
\end{equation}
where
$$c_{k_+} = \frac{1}{k_+!}\frac{1}{(1+\tau)^{k_+(k_+-1)/2}}.$$
Now, (\ref{symm2}) is written as
\begin{eqnarray*}
  & &c_{k_+}\prod_{0<i \neq j}\frac{\xi_j - \xi_i}{p+q\xi_i\xi_j -\xi_i}\prod_{i>0}\frac{1}{(1-\xi_i)(\xi_i-\tau)}
  \prod_{i>0}
  \Big(\frac{\xi_i^{x+1}e^{\varepsilon(\xi_i)t}(\xi_i-\tau)}{\xi_i^2-\tau}\Big)
 \prod_{0<i<j}\frac{1+\tau -(\xi_i+\xi_j)}{\tau -\xi_i\xi_j} \\
 & &\hspace{0.3cm}=\tilde{c}_{k_+}\det\Big(\frac{1}{p+q\xi_i\xi_j -\xi_j}\Big)_{1\leq i,j \leq k_{+}}\prod_{i>0}\Big(\frac{\xi_i^{x+1}e^{\varepsilon(\xi_i)t}(\xi_i-\tau)}{\xi_i^2-\tau}\Big)
 \prod_{0<i<j}\frac{1+\tau -(\xi_i+\xi_j)}{\tau -\xi_i\xi_j} \\
  & &\hspace{0.3cm}=\tilde{c}_{k_+} \det\big(K_+(\xi_i,\xi_j) \big)_{1\leq i,j \leq k_+}\prod_{0<i<j}\frac{1+\tau -(\xi_i+\xi_j)}{\tau -\xi_i\xi_j} \label{identity}
\end{eqnarray*}
where
\begin{equation*}
 \tilde{c}_{k_+} = c_{k_+}\cdot(-1)^{-k_+}(pq)^{-k_+(k_+-1)/2}q^{k_+}
\end{equation*}
and
\begin{equation*}
 K_+(\xi,\xi') = \frac{\xi'^{x}e^{\varepsilon(\xi')t}}{p+q\xi\xi' -\xi}\cdot(\xi'-\tau)\cdot\frac{\xi'}{(\xi')^2-\tau}.
\end{equation*}
Likewise, for the second bracket, let $\tilde{\xi_i} = \xi_{-i}^{-1}$. Then
\begin{eqnarray*}
\prod_{i<j<0}f(\xi_i, \xi_j)&=& \prod_{0<j<i} \frac{\tilde{\xi_j}^{-1} -\tilde{\xi_i}^{-1}}{p + q\tilde{\xi_i}^{-1}\tilde{\xi_j}^{-1} - \tilde{\xi_i}^{-1}}
 = \prod_{0<j<i}\frac{\tilde{\xi_i} -\tilde{\xi_j}}{p\tilde{\xi_i}\tilde{\xi_j} + q - \tilde{\xi_j}}\\
  & =&\prod_{0<i<j}\frac{\tilde{\xi_j} -\tilde{\xi_i}}{q+p\tilde{\xi_i}\tilde{\xi_j} - \tilde{\xi_i}}
\end{eqnarray*}
and
\begin{eqnarray*}
  & &\varphi_-(k_-,\xi_-) = \frac{\tilde{\xi_1}^{-1}\tilde{\xi_2}^{-3}\cdots\tilde{\xi_{k_-}}^{-(2k_--1)}}{\big(\tau^{k_-} - (\tilde{\xi_{1}}\cdots\tilde{\xi_{k_-}})^{-2} \big)\big(\tau^{k_--1}-(\tilde{\xi_{2}}\cdots\tilde{\xi_{k_-}})^{-2}\big)\cdots
  \big(\tau - (\tilde{\xi_{k_-}})^{-2} \big)}\\
  &=& \frac{\tau^{\frac{k_-(k_-+1)}{2}-\frac{k_-(k_-+1)}{2}}\cdot\tilde{\xi_1}^{-2}\tilde{\xi_2}^{-4}\cdots\tilde{\xi_{k_-}}^{-2k_-}
  \cdot(\tilde{\xi_1}\cdots\tilde{\xi_{k_-}})}
  {\big(\tau^{k_-} - (\tilde{\xi_{1}}\cdots\tilde{\xi_{k_-}})^{-2} \big)\big(\tau^{k_--1}-(\tilde{\xi_{2}}\cdots\tilde{\xi_{k_-}})^{-2}\big)\cdots \big(\tau - (\tilde{\xi_{k_-}})^{-2} \big)}\\
  &=&\frac{\tau^{-\frac{k_-(k_-+1)}{2}}\cdot(\tilde{\xi_1}\cdots\tilde{\xi_{k_-}})}
  {\big((\tilde{\xi_{1}}\cdots\tilde{\xi_{k_-}})^{2}-\frac{1}{\tau^{k_-}}\big)
  \big((\tilde{\xi_{2}}\cdots\tilde{\xi_{k_-}})^{2}-\frac{1}{\tau^{k_--1}}\big)\cdots\big((\tilde{\xi_{k_-}})^{2} -\frac{1}{\tau}\big)}.
\end{eqnarray*}
Using Lemma 3.2 again, the symmetrization of
\begin{equation*}
 \prod_{0<i<j}\frac{\tilde{\xi_j} -\tilde{\xi_i}}{q+p\tilde{\xi_i}\tilde{\xi_j} - \tilde{\xi_i}}
 \frac{\tau^{-\frac{k_-(k_-+1)}{2}}\cdot(\tilde{\xi_1}\cdots\tilde{\xi_{k_-}})}
  {\big((\tilde{\xi_{1}}\cdots\tilde{\xi_{k_-}})^{2}-\frac{1}{\tau^{k_-}}\big)
  \big((\tilde{\xi_{2}}\cdots\tilde{\xi_{k_-}})^{2}-\frac{1}{\tau^{k_--1}}\big)\cdots\big((\tilde{\xi_{k_-}})^{2} -\frac{1}{\tau}\big)}
 \end{equation*}
 is
  \begin{equation*}
  c_{k_-}\prod_{0< i\neq j}\frac{\tilde{\xi_j} - \tilde{\xi_i}}{q + p\tilde{\xi_i}\tilde{\xi_j} - \tilde{\xi_i}}\prod_{i>0}^{k_-}\Big(\tilde{\xi_i}\cdot \frac{1}{\tilde{\xi_i}^2-\tau^{-1}}\Big)\prod_{0<i<j}\frac{1+\tau^{-1} - (\tilde{\xi_i} + \tilde{\xi_j})}{\tau^{-1} - \tilde{\xi_i}\tilde{\xi_j}}
\end{equation*}
where
$$c_{k_-}= \tau^{-\frac{k_-(k_-+1)}{2}}\frac{1}{k_-!}\frac{1}{(1+\tau^{-1})^{k_-(k_--1)/2}},
  $$
  and recalling that  $\tilde{\xi_i} = \xi_{-i}^{-1}$, we obtain the symmetrization of the second bracket
\begin{equation}
c_{k_-}\prod_{0>i \neq j}\frac{\xi_j - \xi_i}{p+q\xi_i\xi_j -\xi_i}\prod_{i<0}^{-k_-}\Big(\frac{\xi_i^{x-1}e^{\varepsilon(\xi_i)t}}{1-\xi_i}\frac{1}{\xi_i^{-2}-1/\tau}\Big)
 \prod_{j<i<0}\frac{1+1/\tau -(\xi_i^{-1}+\xi_j^{-1})}{1/\tau -(\xi_i\xi_j)^{-1}} .\label{negative}
\end{equation}
Using Lemma 3.1, (\ref{negative}) is expressed as
\begin{equation*}
  \tilde{c}_{k_-} \det\big(K_-(\xi_i,\xi_j) \big)_{-k_-\leq i,j \leq -1}\prod_{j<i<0}\frac{1+\tau^{-1} -(\xi_i^{-1}+\xi_j^{-1})}{\tau^{-1} -\xi_i^{-1}\xi_j^{-1}}
\end{equation*}
where
\begin{equation*}
 \tilde{c}_{k_-} = c_{k_-}\cdot(-1)^{-k_-}(pq)^{-k_-(k_--1)/2}q^{k_-}
\end{equation*}
and
\begin{equation*}
 K_-(\xi,\xi') = \frac{\xi'^{x}e^{\varepsilon(\xi')t}}{p+q\xi\xi' -\xi}\cdot(\xi'-\tau)\cdot\frac{(\xi')^{-1}}{(\xi')^{-2}-\tau^{-1}}.
\end{equation*}
Now, we summarize our result. Recall the definition of $f(\xi_i,\xi_j)$ in Section 2. Then we have
\begin{eqnarray}
 & &\mathbb{P}_{2\mathbb{Z}-1}(X_m(t)\leq x)\label{final}\\
  &=&\nonumber\sum_{k_\pm\geq 0}c_{m,k_{\pm}}\tilde{c}_{k_-}\tilde{c}_{k_+}\int_{\mathcal{C}_r^{k_-}}\int_{\mathcal{C_R}^{k_+}}\prod_{i<0,j>0}f(\xi_i,\xi_j)\cdot G_-(\xi_-)G_+(\xi_+)\prod_id\xi_i
\end{eqnarray}
where
\begin{eqnarray*}
 G_-(\xi_-) &=& \det\big(K_-(\xi_i,\xi_j) \big)_{-k_-\leq i,j \leq -1}\prod_{j<i<0}\frac{1+\tau^{-1} -(\xi_i^{-1}+\xi_j^{-1})}{\tau^{-1} -\xi_i^{-1}\xi_j^{-1}},\\
 G_+(\xi_+) &=&  \det\big(K_+(\xi_i,\xi_j) \big)_{1\leq i,j \leq k_+}\prod_{0<i<j}\frac{1+\tau -(\xi_i+\xi_j)}{\tau -\xi_i\xi_j}
 \end{eqnarray*}
 and
 \begin{eqnarray*}
  c_{m,k_{\pm}}\tilde{c}_{k_+} \tilde{c}_{k_-} &=& (-1)^{\frac{m+1}{2}-k_+}\cdot\frac{1}{k_+!k_-!}\cdot\tau^{\big(\frac{m^2-1}{8} + \frac{k_+(k_++1)}{2} - \frac{m+1}{2}k_+\big)}\cdot q^{\frac{k(k+1)}{2}} \times \\
  & & p^{-k_+(k_+-1)/2}\cdot q^{-k_-(k_--1)/2}\cdot
  \left[
                                                              \begin{array}{c}
                                                                k-1 \\
                                                                \frac{m-1}{2} + k_- \\
                                                              \end{array}
                                                            \right]_\tau.
 \end{eqnarray*}

 \subsection{The one-sided alternating initial condition}
Let us assume the initial condition $Y'= \{2n-k_0: n \in \mathbb{N}\}$ for a fixed $k_0 \in \mathbb{Z}.$ In this case we start with (3) in \cite{TW3},
\begin{equation}
 \mathbb{P}_{Y}(x_m(t) \leq x) = \sum_{k \geq 1}c_{m,k}\sum_{\substack{S \subset Y,\\|S|=k}}\tau^{\sigma(S,Y)}\int_{\mathcal{C_R}}\cdots \int_{\mathcal{C_R}}I(x,k,\xi)\prod_i\xi_i^{-s_i}\prod_i d\xi_i\label{basic}
\end{equation}
where
\begin{equation*}
 c_{m,k} = q^{k(k-1)/2}(-1)^{m}\tau^{m(m-1)/2-km}\left[
                                                            \begin{array}{c}
                                                              k-1 \\
                                                              m-1 \\
                                                            \end{array}
                                                          \right]_{\tau}
\end{equation*}
for convergence of geometric series.
 By using the same procedure as the subsection 1 and the identity in Lemma 3.2 we can obtain $\mathbb{P}_{Y'}(x_m(t) \leq x)$. We give the integral formula without the detailed derivation.
\begin{eqnarray}
  & & \mathbb{P}_{Y'}(x_m(t) \leq x) = (-1)^m\sum_{k \geq m}\frac{\tau^{(k-m)(k-m+1)/2}}{(1+\tau)^{k(k-1)}k!}
  \left[
                                                                                                             \begin{array}{c}
                                                                                                               k-1 \\
                                                                                                               k-m \\
                                                                                                             \end{array}
                                                                                                           \right]_{\tau} \label{oneside} \\
      & & \nonumber \hspace{0.1cm}\times\int_{\mathcal{C_{R}}}\cdots  \int_{\mathcal{C_{R}}}  \prod_{i\neq j}\frac{\xi_j-\xi_i}{p+q\xi_i\xi_j -\xi_i}\prod_{i}\frac{\xi_i^{x+k_0}e^{t\varepsilon (\xi_i)}}{(1-\xi_i)(\xi_i^2-\tau)}\prod_{i<j}\frac{1+\tau - (\xi_i +\xi_j)}{\tau -\xi_i\xi_j}\prod_i d\xi_i.
  \end{eqnarray}
This is to be compared with (\ref{step}) with step initial condition and (\ref{final}) with the two-sided alternating initial condition. (\ref{oneside}) has a product term  in the integrand compared with (\ref{step}) and is described by the integral only over large contours with the integrand $G_+(\xi)$ (when $k_0 =1$) compared with (\ref{final}).
\begin{remark}
 One may consider more generalized \textit{periodic} initial condition that sites in $k\mathbb{Z}\hspace{0.1cm} (k\in \mathbb{N}, k \geq 2)$ are occupied and all other sites are empty. The discrete TASEP with this initial condition was studied in \cite{Borodin}. We confirmed by using a computer that we do not have a combinatorial identity in a simple form when $k \geq 3.$ That is, currently, we have two special deterministic initial conditions, that is, the step initial condition and the alternating initial condition, which are associated with combinatorial identities.
\end{remark}

\noindent\textbf{Acknowledgement} \\
The author is grateful to Craig A. Tracy for suggesting this problem and invaluable comments on this work and thanks anonymous referees for useful comments. {This work was supported in part by National Science Foundation through the grant DMS-0906387.}

\end{document}